\documentclass[%
 reprint,
 amsmath,amssymb,
 aps,floatfix,superscriptaddress,
pra,
nofootinbib]
{revtex4-1}

\usepackage{graphicx}
\usepackage{dcolumn}
\usepackage{bm}
\usepackage{epstopdf}
 \usepackage[usenames,dvipsnames]{pstricks}
 \usepackage{epsfig}
 \usepackage{pst-grad} 
 \usepackage{pst-plot} 
 \usepackage[space]{grffile} 
 \usepackage{etoolbox} 
 \makeatletter 
 \patchcmd\Gread@eps{\@inputcheck#1 }{\@inputcheck"#1"\relax}{}{}
 \makeatother
\usepackage{booktabs}
\usepackage{boxhandler}
\usepackage{subcaption}
\usepackage{braket}
\usepackage[cmintegrals]{newtxmath}
\usepackage{multirow}
\usepackage{enumitem}

\usepackage{amsmath,amssymb,amsfonts,amsthm}

\usepackage{algorithmic}
\makeatletter
\let\MYcaption\@makecaption
\makeatother

\usepackage[font=footnotesize]{subcaption}
\newtheorem{thm}{Theorem}
\newtheorem{corollary}{Corollary}
\newtheorem{lemma}{Lemma}

\makeatletter
\let\@makecaption\MYcaption
\makeatother
\usepackage{url}
\usepackage{braket}
\usepackage{textcomp}
\usepackage{multirow}
\usepackage{xcolor}

\begin{document}

\title{The superadditivity effects of quantum capacity decrease with the dimension for qudit depolarizing channels}

\author{Josu {Etxezarreta Martinez}}
\email[Corresponding author: ]{jetxezarreta@tecnun.es}
    \affiliation{ Tecnun - University of Navarra, 20018 San Sebastian, Spain}
\author{Antonio deMarti iOlius}
    \affiliation{ Tecnun - University of Navarra, 20018 San Sebastian, Spain}
\author{Pedro M. Crespo} 
    \affiliation{ Tecnun - University of Navarra, 20018 San Sebastian, Spain}
\date{\today}

\begin{abstract}
Quantum channel capacity is a fundamental quantity in order to understand how good can quantum information be transmitted or corrected when subjected to noise. However, it is generally not known how to compute such quantities, since the quantum channel coherent information is not additive for all channels, implying that it must be maximized over an unbounded number of channel uses. This leads to the phenomenon known as superadditivity, which refers to the fact that the regularized coherent information of $n$ channel uses exceeds one-shot coherent information. In this article, we study how the gain in quantum capacity of qudit depolarizing channels relates to the dimension of the systems considered. We make use of an argument based on the no-cloning bound in order to proof that the possible superadditive effects decrease as a function of the dimension for such family of channels. In addition, we prove that the capacity of the qudit depolarizing channel coincides with the coherent information when $d\rightarrow\infty$. We also discuss the private classical capacity and obain similar results. We conclude that when high dimensional qudits experiencing depolarizing noise are considered, the coherent information of the channel is not only an achievable rate but essentially the maximum possible rate for any quantum block code.
\end{abstract}

\maketitle

\section{Introduction}
Classical communications were revolutionized when Claude Shannon introduced the noisy-channel coding theorem in his groundbreaking work \textit{A Mathematical Theory of Communication} \cite{shannon}. In such theorem, Shannon introduced the concept of channel capacity, which refers to the maximum coding rate for which asymptotically error-free communications are possible over a noisy channel. The consequences to this result are momentous since it establishes the limit, in terms of rate, for which error correction makes sense and, thus, the target that coding theorists should seek when designing their codes. The computation of such quantity results to be simple due to the fact that the classical mutual information is additive, implying that the regularization over $n$ channel uses needed to compute the capacity of the channel results in a single-letter formula, i.e. in the optimization of such quantity over a single use of the channel \cite{shannon}.

The development of quantum information theory followed the steps of Shannon, introducing the concept of quantum channel capacity similarly to its classical counterpart, i.e. the maximum quantum coding rate for comunication/correction (note that in the quantum setting the noise can arise from temporal evolution) with error rates vanishing asymptotically when quantum information is subjected to noise. In general, the computation of the quantum channel capacity, $C_\mathrm{Q}$, is based on the following regularization \cite{lsd1,lsd2,lsd3,lsd4,wildeQIT}:

\begin{equation}\label{eq:cqregedf}
C_\mathrm{Q}(\mathcal{N}) = \lim_{n\rightarrow \infty} \frac{1}{n} Q_{\mathrm{coh}}(\mathcal{N}^{\otimes n}),
\end{equation}
where $\mathcal{N}$ denotes the quantum channel and $Q_{\mathrm{coh}}$ refers to the channel coherent information defined as

\begin{equation}\label{eq:cohdef}
\begin{split}
Q_{\mathrm{coh}}(\mathcal{N}) &= \max_{\rho} I_{\mathrm{coh}}(\mathcal{N},\rho) \\&= \max_{\rho} S(\mathcal{N}(\rho)) - S(\mathcal{N}^c(\rho)),
\end{split}
\end{equation}
with $I_{\mathrm{coh}}(\mathcal{N},\rho)$ the channel coherent information when state $\rho$ is the input, $S$ the von Neumann entropy and $\mathcal{N}^c$ is a complementary channel to the environment.

However, in stark contrast to its classical counterpart, the channel coherent information has been proven not to be additive in general \cite{wildeQIT,super1,super2,super3,unbounded}, implying that the regularization in equation \eqref{eq:cqregedf} involves optimizing over an infinite parameter space.  Given two arbitrary quantum channels $\mathcal{N}_1$, $\mathcal{N}_2$, the most one can say about the coherent channel information of the parallel channel $\mathcal{N}_1\otimes \mathcal{N}_2$ is $Q_\mathrm{coh}(\mathcal{N}_1\otimes \mathcal{N}_2)\geq Q_\mathrm{coh}(\mathcal{N}_1)+Q_\mathrm{coh}(\mathcal{N}_2)$. When strict inequality holds, the channels are said to exhibit superaditivity, otherwise are said to have additive coherent information \cite{siddhu}. Explicit examples of superadditivity have been found for several classes of quantum channels \cite{wildeQIT,unbounded,siddhu,super1,super2,super3,super4,super5,super6,super7,super8,xzzx,neuralSuper,geneticSuper}. Importantly, the non-additivity effects of quantum capacity arise as a result of entanglement in the input state of the channel since state coherent information is additive for unentangled input states, i.e. $I_\mathrm{coh}(\mathcal{N}^{\otimes 2}, \rho\otimes\sigma) = I_\mathrm{coh}(\mathcal{N}, \rho) + I_\mathrm{coh}(\mathcal{N}, \sigma)$ \cite{unbounded}. This implies that entanglement is a resource that may protect quantum information from noise in a more efficient way than what it is classically possible.

Therefore, an important question to be answered is what types of channels have additive channel coherent information so that their capacity reduces to single-letter expressions, i.e., $C_\mathrm{Q}(\mathcal{N}) =\max_{\rho} I_{\mathrm{coh}}(\mathcal{N},\rho) $. At the time of writing, quantum channels with additive channel coherent information belong to the classes of degradable \cite{wildeQIT,degOrig,degDep,watanabe,ftvqc}, conjugate degradable \cite{watanabe,cdeg} and less noisy than the environment \cite{watanabe} channels. The quantum capacities of antidegradable, conjugate antidegradable and entanglement-binding channels are also single-letter characterized, but they are equal to zero \cite{wildeQIT,unbounded,watanabe,entBind,ftvqc}. Recently, examples of quantum channels (the platypus, multi-level amplitude damping and resonant multi-level amplitude damping channels), showing additivity while not being degradable have been found \cite{platypus,mad,remad}.

The depolarizing channel is a widely used quantum channel model in order to describe the noise that quantum information experiences \cite{josuchannels}. This channel is characterized by the depolarizing probability, $p$, and its quantum channel capacity is still unknown even if it is the simplest and most symmetric nonunitary quantum channel. In general, $d$-dimensional depolarizing channels (those acting on $d$-dimensional quantum states referred as qudits) are antidegradable for $p \geq \frac{d}{2(d+1)}$, while they do not belong to any of the classes of channels previously mentioned for $p < \frac{d}{2(d+1)}$ \cite{depAnti}. Several upper bounds on the quantum capacity of $d$-dimensional depolarizing channels for the non-trivial parameter region have been derived \cite{UB1,UB2,UB3,UB4,UB5,UB6}. However, the quantum capacity of the family of $d$-dimensional depolarizing channels remains a mistery for such region.

In this article, we study how the potential superadditivity effects of the quantum channel capacity, in qudits per channel use units, relate to the dimension of the depolarizing channel. Specifically, we want to observe which is the extra coding rate achievable due to superadditivity when logical qudits are encoded by physical qudits. We provide an argument based on the no-cloning bound in order to study how the quantum capacity gain (defined as the difference between the quantum capacity and the channel coherent information) caused by potential coherent information superadditivity relates to the dimension of the depolarizing channel. We conclude that such possible capacity gain is a monotonically decreasing function with the dimension and, thus, that the superadditive effects are less and less important when the dimension of the depolarizing channels increases. In addition, we determine that for the extremal case in which the dimension of the system is let to grow indefinitely (in the limit where the qudit becomes a quantum oscillator, i.e., a bosonic mode \cite{qudits}), the depolarizing channel capacity coincides with the channel coherent information. We also relate the obtained results with the private capacity of qudit depolarizing channels concluding that such information theoretic quantity behaves in a similar way as the quatum channel capacity.

\section{Qudit depolarizing channels}\label{sec:depolChann}
The $d$-dimensional or qudit depolarzing channel, $\Lambda_p^d: \mathcal{H}_d\rightarrow \mathcal{H}_d$, is the completely-positive, trace preserving (CPTP) map defined as \cite{UB1,UB6,dep1,dep2,dep3}
\begin{equation}\label{eq:depol}
\Lambda_p^d(\rho) = (1-p)\rho + p\mathrm{Tr}(\rho)\frac{I_d}{d},
\end{equation}
where the density matrices $\rho$ are the so-called qudits or quantum states operating over a $d$-dimensional Hilbert space $\mathcal{H}_d$, $I_d/d$ refers to the maximally mixed state of dimension $d$ and $p\in[0,1]$ refers to the depolarizing probability. Consequently, the operation of the qudit depolarizing channel leaves the state uncorrupted with probability $1-p$ while transforming it to the maximally mixed state with probability $p$.

The depolarizing channel has a central role in modeling quantum noise in the theory of quantum information \cite{josuchannels}. Importantly, depolarizing channels can be efficiently simulated as an stochastic noise map by classical means since they fulfill the Gottesman-Knill theorem \cite{josuchannels,GKthm}. This implies that, for example, the performance of quantum error correction codes, key for fault-tolerant quantum computing and communications, can be effectively assesed by traditional methods. Furthermore, Clifford twirling an arbitrary $d$-dimensional CPTP noise map results in a qudit depolarizing channel \cite{josuchannels,twirldep}. Twirling is extensively used in quantum information theory for studying the average effects of a general noise map by mapping them to more symmetric versions of themselves \cite{josuchannels,depAnti,twirl1,twirl2,twirl3}. The twirled channel is obtained by averaging the action of the map over a set of unitaries.  Moreover, the following lemma \cite{twirl1} implies that error correction codes for arbitrary noise maps can be designed by constructing them to correct a twirled map.
\begin{lemma}
\textit{Any correctable code for the twirled channel $\bar{\mathcal{N}}$ is a correctable code for the original channel $\mathcal{N}$ up to an additional unitary correction.}
\end{lemma}

Hence, the depolarizing channel is not only interesting because of its nice properties, but also as error correction codes can be designed by using it. The depolarizing parameter and the parameters of the original channel are related in a specific way as a result of the twirl (see \cite{josuchannels,TVQC} for specific details on the qubit case). Notably, twirling channels into Pauli channels, whose symmetric version is the depolarizing channel, has recently been used for the quantum error mitigation technique named Probabilisitc Error Cancellation (PEC) \cite{PEC}. 

Consequently, studying the achievable rates for the different quantum information theoretical tasks over depolarizing channels is of the outmost importance. Studying the different capacities of such family of channels is also interesting from the point of view of quantum information theory too since the capacities of twirled channels lower bound the capacities of the channels from which they originated \cite{depAnti} and, thus, interesting lower bounds on the achievable rates of general channels might be obtained.

The channel coherent information, $Q_{\mathrm{coh}}$, defined in equation \eqref{eq:cohdef}, for qudit depolarizing channels is \cite{UB1}
\begin{equation}\label{eq:coherentQubits}
\begin{split}
& Q_{\mathrm{coh}}(\Lambda_p^d) =\max \left\lbrace 0, \log_2{d}\right. \\ & + \left(1-p\frac{d^2 - 1}{d^2}\right)\log_2{\left(1-p\frac{d^2 - 1}{d^2}\right)} \\ & \left. + p\frac{d^2 - 1}{d^2}\log_2{\left(\frac{p}{d^2}\right)}\right\rbrace,
\end{split}
\end{equation}
with units of qubits per channel use. It provides a lower bound for the quantum channel capacity, $C_\mathrm{Q}(\mathcal{N})\geq Q_{\mathrm{coh}}(\mathcal{N})$.
Note that by changing the $\log_2$ in the above expression by $\log_d$, the units of $ Q_{\mathrm{coh}}(\Lambda_p^d)$ are qudits per channel use. The reason to consider this units is that we are interested in studying the logical qudits per physical qudits, i.e. coding rate, that can be achieved for a qudit error correction scheme, and not the amount of logical qubits that can be encoded by means of qudits. For the sake of notation we will denote the channel coherent information in such units by $Q^d_{\mathrm{coh}}(\Lambda_p^d)$.

Recall that for $p< \frac{d}{2(d+1)}$ the channel does not belong to any of the classes with proven additive channel coherent information \cite{degDep}, implying that the quantum channel capacity is not known and may exhibit superadditivity gains. In fact, these gains have been obtained in previous works \cite{wildeQIT, super1,super2,neuralSuper}. Several techniques have been developed in order to obtain upper bounds for the quantum channel capacity of  $d$-dimensional depolarizing channels \cite{UB1,UB2,UB3,UB4,UB5}. Each of those upper bounds are tighter depending on the region of depolarizing probability considered in $p\in\left[0,\frac{d}{2(d+1)}\right]$. The tightest upper bound is usually obtained by using the fact that the convex hull of the upper bounds is itself an upper bound \cite{UB2}. However, for the purposes of this work, we will consider the so called no-cloning bound, $Q_{\mathrm{nc}}$. The no-cloning bound on quantum capacity is based on combining Cerf's no-cloning bounds \cite{nc2} and the degradable extension technique of \cite{UB6}. Cerf's results lay on the no-cloning theorem\footnote{A unitary operator that perfectly copies arbitrary quantum states cannot be constructed.} of quantum mechanics for determining that Pauli channels (depolarizing channels are an specific instance of this) cannot have a positive capacity under certain conditions. By using this result, the bound can be obtained by the techniques in \cite{UB6}. A proof for this can be found in \cite{depAnti}. The no-cloning bound upper bounds the quantum channel capacity of qudit depolarizing channels as \cite{depAnti,UB2,nc1,nc2,nc3}
\begin{equation}\label{eq:nocloningQubits}
C_\mathrm{Q}(\Lambda_p^d)\leq Q_{\mathrm{nc}}(\Lambda_p^d) = \left(1-2p\frac{d+1}{d}\right)\log_2{d},
\end{equation}
with units of qubits per channel use. Note that the expresion of $Q_{\mathrm{nc}}(\Lambda_p^d)$ in qudits per channel use reduces to
\begin{equation}\label{eq:nocloningQudits}
Q_{\mathrm{nc}}^d(\Lambda_p^d) = \left(1-2p\frac{d+1}{d}\right).
\end{equation}

\section{Superadditivity gain}\label{sec:supergain}
As explained in the previous section, the potential superadditive nature of the coherent information may lead to quantum channel capacities that are higher than the one-shot channel coherent information. In other words, there exists a gain in quantum channel capacity if several quantum channel uses are considered. Remarkably, it has been proven that even an unbounded number of channels uses may be required for this effect to happen \cite{unbounded}. In order to quantify this gain we define the superadditivity gain, $\xi$, as
\begin{equation}\label{eq:supergain}
\xi(\mathcal{N}) = C_\mathrm{Q}(\mathcal{N}) - Q_\mathrm{coh}(\mathcal{N}),
\end{equation}
which gives the additional qubits per channel that the channel capacity has when compared the achievable rate of the channel coherent information. Clearly, if the coherent information of the channel is additive, then  $\xi(\mathcal{N}) =0$. Knowledge about the quantum channel capacity is needed in order to compute the superadditivity gain in equation \eqref{eq:supergain} and, as stated before, the quantum capacity of qudit depolarizing channels is still unknown. However, upper bounds on such quantity can be obtained using the upper bounds derived in \cite{UB1,UB2,UB3,UB4,UB5,UB6}. For the purposes of this work we will upper bound the superadditivity gain by using the no-cloning bound as
\begin{equation}\label{eq:supergainnocloning1}
\xi_\mathrm{nc}(\Lambda_p^d) = Q_{\mathrm{nc}}(\Lambda_p^d) - Q_\mathrm{coh}(\Lambda_p^d)\geq \xi(\Lambda_p^d).
\end{equation}
The units in the above expression are qubits per channel use. However, we will study the capacity gain with qudits per channel use units in order to have a fair comparison of the extra capacity that is obtained via superadditive effects. In this way, we will be able to see how many more qudits per channel use can be potentially obtained due to superadditive effects, which is more fair to compare those effects for different dimensions, since operating in more dimensions trivially implies that more information (in terms of qubits) can be encoded in a single quantum state.  For example, consider $d_1<d_2$ and assume that their superadditivity gains in qudits per channel use (coding rate) for both cases is the same. That is, $\xi(\Lambda_p^{d_1})=\xi(\Lambda_p^{d_2})=g$. However, these gains become  $g\log_2(d_1) < g\log_2(d_2)$ when expressed in qubits per channel use, making the impression that the capacity of for $d_2$ increases more. Note that whenever qudit error correction codes are constructed, their coding rate will have logical qudits per physical qudits units, implying that the extra rate obtained via superadditivity should be quantified in such terms.

Therefore, in what follows, the units of the superadditive gains will be given in qudits per channel use, that is
\begin{equation}\label{eq:supergainnocloning}
\xi_\mathrm{nc}(\Lambda_p^d) = Q^d_{\mathrm{nc}}(\Lambda_p^d) - Q^d_\mathrm{coh}(\Lambda_p^d)\geq \xi(\Lambda_p^d).
\end{equation}

\section{Superadditivity effects of quantum capacity decrease as a function of the dimension}\label{eq:results}
We now provide the main result of this article. 
\begin{thm}\label{thm:decrease}
\textit{Let $d_l$ be an arbitrary positive integer higher than $2$ and $p_0^{d_l}\in \mathbb{R}$ defined as
\begin{equation}\label{eq:zerodep}
p_0^{d_l} = \min_p\left(\left\lbrace p\in\left(0,\frac{d_l}{2(d_l+1)}\right) : Q_\mathrm{coh}^{d_l}(\Lambda_p^{d_l}) = 0 \right\rbrace\right).
\end{equation}
That is, $p_0^{d_l}$ is the smallest depolarizing probability that makes the coherent information of the $d_l$-dimensional depolarizing channel equal to zero. 
Then, for any depolarizing probability $p$ in the range $p\in(0,p_0^{d_l})$, the superadditivity gain, 
 $\xi_\mathrm{nc}(\Lambda_p^d)$, in qudits per channel use units is a monotonically decreasing function of the channel dimension, $d$, for $d\geq d_l$.
}
\end{thm}
\begin{proof}
To prove the theorem, we must prove that 
\begin{equation}
\frac{\partial \xi_\mathrm{nc}(\Lambda_p^d)}{\partial d} < 0, \quad \forall p\in\left(0,p_0^{d_l}\right).
\end{equation}
Thus, the derivative of $\xi_\mathrm{nc}(\Lambda_p^d)$ over the dimension in the range $p\in\left(0,p_0^{d_l}\right)$
\begin{equation}
\begin{split}
\frac{\partial \xi_\mathrm{nc}(\Lambda_p^d)}{\partial d} &= -1 - 4\frac{p}{d} + 4p\frac{(d^2-1)}{d^3}  - p\frac{(d^2-1)\log_2{\left(\frac{p}{d^2}\right)}}{d^2\log_2{d}} \\ & -\frac{(1-p\frac{d^2 - 1}{d^2})\log_2{\left(1-p\frac{d^2 - 1}{d^2}\right)}}{\log_2(d)} \\ & = - 4\frac{p}{d} + 4p\frac{(d^2-1)}{d^3} - Q_{\mathrm{coh}}^d(\Lambda_p^d)<0.
\end{split}
\end{equation}
The last inequality follows from the fact that $4\frac{p}{d} > 4p\frac{(d^2-1)}{d^3}, \forall d$ (this inequality reduces to $\frac{1}{d}> \frac{1}{d}-1$ which is true for all $d>0$), and the fact that $\forall d\geq d_l$, $Q_{\mathrm{coh}}^d(\Lambda_p^d)\geq 0$, since $p_0^d$ increases with $d$ and we are considering the range $p\in\left(0,p_0^{d_l}\right)$.
\end{proof}

Figure \ref{fig:maxdecrease_p} graphically shows the results of this theorem. It plots the no-cloning superadditivity gain versus depolarizing probability, $p$, for four different $d_l$ dimensions. For a given $d_l$, the vertical dashed lines give the value of the corresponding $p_0^{d_l}$.
\begin{figure}[!ht]
\centering
\includegraphics[width=\linewidth]{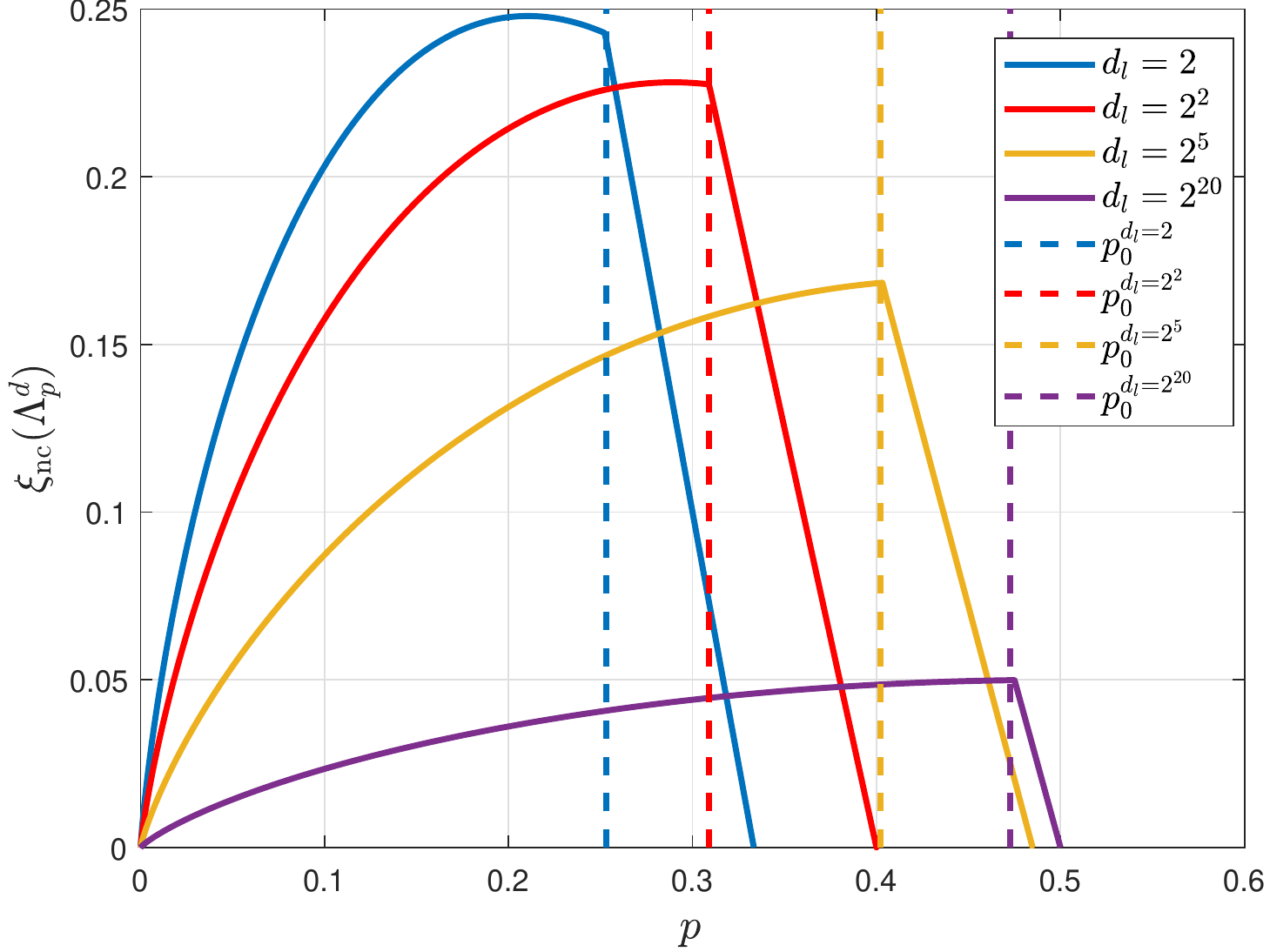}
\caption{\textbf{No-cloning superadditivity gain as a function of depolarizing probability $p$.} Channel dimensions $d\in \{2,2^2,2^5,2^{20}\}$ are plotted.}
\label{fig:maxdecrease_p}
\end{figure}

Note the result of Theorem \ref{thm:decrease} states that for an initial dimension $d_l$, the no-cloning superadditive gain, $\xi_\mathrm{nc}(\Lambda_p^d)$ is a decreasing function with respect to the dimension $d\geq d_l$ in the depolarizing probability range  $p\in(0,p_0^{d_l})$. It is noteworthy that the result of the theorem can be extended to a non-tirival region where the coherent information vanishes. However, since the point of maximum potential superadditivity lays in the region considered, expanding the analysis to such parameter space would result in similar conclusions. Additionally, the upper limit of such range, $p_0^{d_l}$, increases with respect to the initial dimension in consideration. This value saturates to $1/2$ when the dimension of the system is left to grow indefinitely since
\begin{equation}\label{eq:p0sat}
\begin{split}
\lim_{d\rightarrow\infty} Q_{\mathrm{coh}}^d(\Lambda_p^d) =& \lim_{d\rightarrow\infty}  \left(1 + \frac{\left(1-p\frac{d^2 - 1}{d^2}\right)\log_2{\left(1-p\frac{d^2 - 1}{d^2}\right)}}{\log_2{d}} \right. \\ & \left. + \frac{p\frac{d^2 - 1}{d^2}\log_2{\left(\frac{p}{d^2}\right)}}{\log_2{d}} \right)  = 1-2p,
\end{split}
\end{equation}
which vanishes at the value of $1/2$.

In this way, by starting with the minimum dimension of a quantum system, i.e. a qubit $d_l=2$, we can always find another initial higher dimension for which the no-cloning superadditive gain decreases in all the range of depolarizing probabilities $p\in(0,1/2)$. For example, see that in Figure \ref{fig:maxdecrease_p} we can change from $d_l=2$ to $d_l=4$ once we reach $p_0^{d_l=2}$, and the gain will still be decreasing for $d>d_l=4$. This can be done each time we reach a particular $d_l$. Thus, we effectively prove that whenever the dimension of the system increases, the room left for superadditive effects in qudits per channel use units decreases. Note also that the region $p\in(0,1/2)$ is actually the only region where superadditivity may happen for every $d$-dimensional depolarizing channels since for $p=0$ there is no noise, implying that $C_\mathrm{Q}^d(\Lambda_p^d) = 1$, while for $p>1/2$ every qudit depolarizing channel is antidegradable since $\lim_{d\rightarrow\infty}d/(2(d+1)) = 1/2$.

Figure \ref{fig:maxdecrease} showcases the decrease of the no-cloning superadditive gain for different depolarizing probabilities $p\in\{0.01,0.05,0.1,0.2,0.25\}$ as a function of the dimension of the system considered.

\begin{figure}[!ht]
\centering
\includegraphics[width=\linewidth]{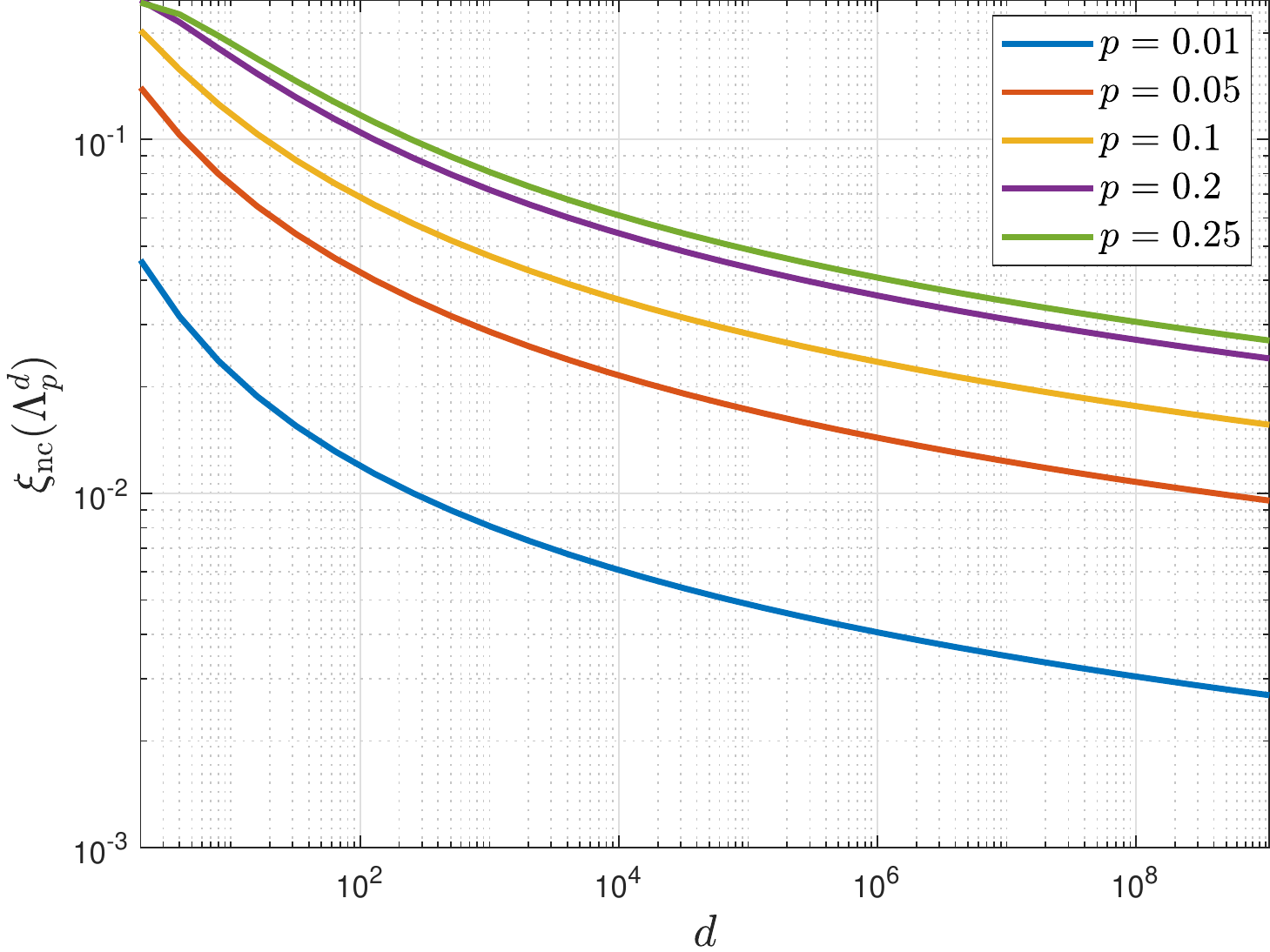}
\caption{\textbf{No-cloning superadditivity gain as a function of dimension and depolarizing probability.} We plot the superadditivity gain in terms of qudits per channel use as a function of the dimension of the depolarizing channel for $p\in\{0.01,0.05,0.1,0.2,0.25\}$.}
\label{fig:maxdecrease}
\end{figure}

Two important conclusions are derived from  Theorem \ref{thm:decrease}, which are clearly appreciated in the above two \texttt{figures}. The first conclusion is that whenever quantum systems of high dimensions are corrupted by the operation of a qudit depolarizing channel, the non-additive behaviour of the coherent information is less relevant. That is, the potential superadditivity gain in terms of qudits per channel use decreases. This is an important result for the depolarizing channel since it implies that for very high dimensional systems, the channel coherent information and the quantum channel capacity will be close together. Note that tighter bounds than the no-cloning bound can be used to bound the superadditivity gain, implying that the actual gain will be much smaller. This yields to the second conclusion which states that for high dimensional systems, the capacity of the depolarizing channel is close to the single-letter coherent information of the channel, that is, one can state that $C_\mathrm{Q}^d(\Lambda_p^d) \approx Q_{\mathrm{coh}}^d(\Lambda_p^d)$. Therefore, we can conclude that for such high dimensional systems, random block codes on the typical subspace of the optimal input (for the one-shot coherent information) will essentially achieve quantum channel capacity \cite{wildeQIT,leditzkyRandom}. This means that the best strategy to achieve the capacity of a depolarizing channels with sufficiently large dimension is by randomly selecting a stabilizer code \cite{wildeQIT}.

We have observed that the superadditive behaviour of coherent information loses importance when the dimensions of the qudit depolarizing channel increase. In particular, in the limit when $d$ is let to be infinite, the qudit becomes a quantum oscillator or bosonic mode \cite{qudits}, and the quantum channel capacity of the $\infty$-dimensional or bosonic depolarizing channel is given by $1-2p$, as it is shown in the following Corollary.
\begin{corollary}\label{col:inftyCap}
\textit{The quantum channel capacity of the $\infty$-dimensional or bosonic depolarizing channel is
\begin{equation}
C_\mathrm{Q}^d(\Lambda_p^\infty) = Q_\mathrm{coh}^d(\Lambda_p^\infty) =1-2p,
\end{equation}
with bosonic modes per channel use units for $p\in[0,1/2]$ and $0$ for $p\in[1/2,1]$.}
\end{corollary}
\begin{proof}
We use a sandwich argument to prove the corollary. We know from equation \eqref{eq:p0sat} that the coherent information of the depolarizing channel has the following asymptotic behaviour in the region $p\in[0,1/2]$
\begin{equation}
C_\mathrm{Q}^d(\Lambda_p^\infty) \geq Q_\mathrm{coh}^d(\Lambda_p^\infty)= \lim_{d\rightarrow \infty}Q_\mathrm{coh}^d(\Lambda_p^d) = 1-2p.
\end{equation}
In addition, if we study the aymptotic behaviour of the no-cloning bound in equation \eqref{eq:nocloningQudits}, then
\begin{equation}
C_\mathrm{Q}^d(\Lambda_p^\infty) \leq \lim_{d\rightarrow \infty}\left(1-2p\frac{d+1}{d}\right) = 1-2p,
\end{equation}
which completes the sandwich and, thus,
\begin{equation}
C_\mathrm{Q}^d(\Lambda_p^\infty) = Q_\mathrm{coh}^d(\Lambda_p^\infty) =1-2p.
\end{equation}

For the complementary region, $p\in[1/2,1]$, we know that this channel is antidegradable. Therefore, the quantum channel capacity vanishes.
\end{proof}

Consequently, it can be seen that the superadditive nature of the coherent information of the qudit depolarizing channel is lost when the dimension of the system is let to grow indefinitely, i.e. $\xi(\Lambda_p^\infty)=0,\forall p$. This result is specially interesting since it is an example of a channel not belonging to the degradable or conjugate degradable classes (the depolarizing channel does not belong to these families of channels), but showing channel coherent information with an additive behaviour. Knowledge about quantum channels presenting additive coherent information but not belonging to the classes of channels known to exhibit additivity is important to obtain a better understanding about the behaviour of quantum channel capacity. In this sense, understanding the structure of particular channels exhibiting additive coherent information may provide hints to understand general classes of channels with such property.

Other implication of Theorem \ref{thm:decrease} and Corollary \ref{col:inftyCap} is that whenever the dimension of the qudit depolarizing channels is big enough, the advantage that utilizing entangled inputs may provide for protecting quantum information losses importance. This comes from the fact that the non-additive effects of coherent information are a result of considering input states that are entangled\footnote{Note that this does not refer to entanglement-assistance, but to the fact that the inputs used over sequential uses of the same channel are entangled among them.} \cite{unbounded}. Since the channel coherent information of is approximately additive for suffficiently high system dimensions, then the use of entangled input states will provide almost no net gain. As entanglement is an expensive resource, this significantly relaxes the required resources for optimal quantum communication/correction over such channels.

To finish with this section, it is important to discuss what happens with the superadditive gain whenever it is considered in qubits per channel use units. As we discussed, we have considered qudits per channel use units since we wanted to study the extra rate achievable due to superadditivity whenever qudit error correction codes are considered, i.e. protecting logical qudits using physical qudits. However, sometimes the information rate in terms of qubits per channel use is also an important thing to study as, for example, when logical qubits want to be encoded by means of qudits \cite{quditQubit,quditQubit2} or when the noise in a system composed by $n$ qubits experiences a depolarizing channel of dimension equal to the Hilbert space of the whole system, i.e. $d=2^n$. The last example would refer to noise that is correlated, that is, noise that cannot be seen as independent noises acting over each of the qubits of the system (See section \ref{sec:implications} for further discussions). In this case, it is more convenient (in terms of calculations) to redefine the superaddivity gain as $\zeta(\mathcal{N}) = C_\mathrm{Q}(\mathcal{N})/Q_\mathrm{coh}(\mathcal{N})$ for obtaining the same results as before. Note for example that it is straightforward to see that $\lim_{d\rightarrow\infty} \zeta(\Lambda_p^d) = (C_\mathrm{Q}^d(\Lambda_p^d)\log_2{d})/(Q_\mathrm{coh}^d(\Lambda_p^d)\log_2{d}) = 1$, implying that the coherent information is additive. Unluckily, this redefinition of the gain poses some problems since it diverges for the region where the coherent information vanishes but the capacity is still strictly positive. However, since Theorem \ref{thm:decrease} considers only the region of positive coherent information and in Corollary \ref{col:inftyCap} such thing is true for the whole region, the same results are obtained. Anyway, we still consider $\xi(\mathcal{N})$ to be the appropriate way to define the superadditivity gain since it is able to capture the non-additivity effects of coherent information for the whole parameter region and, thus, discussed such quantity in terms of qudits per channel use.

\section{Relationship with other capacities}\label{sec:othecaps}
Generally speaking, quantum channels have many other quantum capacities associated with the optimal rate at which some information theoretic task can be performed. Therefore, in this section we discuss the superadditive gain for the classical and private capacities of qudit depolarizing channels.

The classical capacity of a quantum channel, $C_\chi(\mathcal{N})$, is defined as the asymptotically achievable rate of reliable transmission of classical information through the noisy channel \cite{classCap,classCap2,partialorders}. The classical capacity of a quantum channel is geven by the following regularized formula
\begin{equation}\label{eq:classcap}
C_\chi(\mathcal{N}) = \lim_{n\rightarrow\infty} \frac{1}{n} \chi(\mathcal{N}^{\otimes n}),
\end{equation}
where $\chi(\mathrm{N})$ is named the Holevo quantity \cite{partialorders} and it is calculated as
\begin{equation}\label{eq:holevo}
\chi(\mathcal{N}) = \sup_{\rho_{XA}} I(X;B)_{\rho},
\end{equation}
where $\rho_{XA}$ refer to pure classical-quantum states \cite{partialorders} and $I(A;B)_{\rho_{AB}} = S(\rho_A) + S(\rho_B) - S(\rho_{AB})$ is the quantum mutual information \cite{wildeQIT}. The mutual information is evaluated with the state $(\mathbb{I}_X\otimes\mathcal{N})(\rho_{XA})$. For arbitrary channels, the Holevo information is superadditive, implying that the regularization in equation \eqref{eq:classcap} is necessary \cite{wildeQIT,partialorders}. However, it is well known that the Holevo information of qudit depolarizing channels is additive, implying that the classical capacity of such families of channels is equal to the Holevo quantity \cite{dep1}. Therefore, the superadditivity gain of the classical capacity vanishes for all depolarizing probabilities.

The private capacity, $P(\mathcal{N})$, referres to the maximum achievable rate for private transmission of information over a quantum channel with an asymptotically vanishing error rate \cite{lsd3,partialorders,privateCap}. Such quantity can be evaluated as
\begin{equation}\label{eq:privCap}
P(\mathcal{N}) = \lim_{n\rightarrow\infty} \frac{1}{n} P_1(\mathcal{N}^{\otimes n}),
\end{equation}
where the one-shot private information is calculated as
\begin{equation}\label{eq:oneshotPriv}
P_1(\mathcal{N}) = \sup_{\rho_{UA}} I(U;B)_\rho - I(U;E),
\end{equation}
where $\rho_{UA}$ refer to mixed classical-quantum states \cite{partialorders}. The mutual informations are evaluated for states $(\mathbb{I}_U\otimes\mathcal{N}) (\rho_{UA})$ and $(\mathbb{I}_U\otimes\mathcal{N}^c) (\rho_{UA})$, respectively. Private capacity has also shown to be a superadditive quantity \cite{partialorders,superPriv}. Importantly, the private capacity upper bounds the unassisted quantum capacity of a quantum channel \cite{partialorders,lamiWilde,privUBcap}, i.e.
\begin{equation}\label{eq:privUBcap}
P(\mathcal{N})\geq C_\mathrm{Q}(\mathcal{N}),
\end{equation}
which also holds for the one-one shot capacities, i.e. $P_1(\mathcal{N}) \geq Q_\mathrm{coh}(\mathcal{N})$. The upper bound saturates for the class of more capable channels (which includes less noisy and degradable channels) \cite{watanabe}.

Moreover, the no-cloning bound in equation \eqref{eq:nocloningQubits} also upper bounds the private capacity of qudit depolarizing channels. We are unaware of a manuscript including this results and, thus, we provide a proof for it.

\begin{corollary}\label{col:nocloningPriv}
\textit{The no-cloning bound, $Q_{\mathrm{nc}}(\Lambda_p^d)$ is an upper bound for the private quantum capacity of qudit depolarizing channels, i.e.
\begin{equation}
P(\mathcal{N}) \leq Q_{\mathrm{nc}}(\Lambda_p^d) = \left(1-2p\frac{d+1}{d}\right)\log_2{d}.
\end{equation}}
\end{corollary}
\begin{proof}
Note that the $d$-dimensional depolarizing channel is both degradable and antidegradable when $p=\frac{d}{2(d+1)}$. Following the rationale in \cite{depAnti}, we can invoke Smith and Smolin's technique of degradable extensions \cite{UB6} to obtain the upper bound given in the corollary by noting that if the additive extension is degradable, then its coherent information does also upper bound the private capacity of the channel (Theorem 3 in \cite{UB6}).
\end{proof}

Similarly as done with the quantum channel capacity, we will define the normalized private capacity as $P^d(\mathcal{N}) = P(\mathcal{N})/\log_2{d}$. Here, the operational meaning of this quantity will be the number of private bits that can be reliably sent per qubit channel use. Note that, since having a higher dimensional system implies that more classical information can be packed, by normalizing this quantity we can more fairly compare how many extra private bits can be achieved due to superadditivity effects when comparing different dimensional depolarizing channels. 

Using the no-cloning bound upper bound on the private capacity, we extend the results of Theorem \ref{thm:decrease} and Corollary \ref{col:inftyCap} for the private capacity of qudit depolarizing channels. 

\begin{corollary}\label{col:decrasePriv}
\textit{The normalized private capacity superadditivity gain of qudit depolarizing channels, $\xi^P(\Lambda_p^d)$, in units of private bits per two-dimensional channel use is upper bounded by $\xi_\mathrm{nc}(\Lambda_p^d)$, which is a monotonically decreasing function with $d$ for any depolarizing probability, $p$, in the range $p\in(0,p_0^{d_l})$, where $p_0^{d_l}$ is defined as in Theorem \ref{thm:decrease} with $d_l$ an arbitrary positive integer higher than $2$. Therefore, the potential gain that can be obtained from superadditive effects for the private capacity decrease with the dimension of the system.}

\textit{Moreover, the normalized private channel capacity of the $\infty$-dimensional or bosonic depolarizing channel coincides with its quantum capacity and is given by
\begin{equation}
P^d(\Lambda_p^\infty) = C_\mathrm{Q}^d(\Lambda_p^\infty) = 1-2p,
\end{equation} 
with private bits per two-dimensional channel use units for $p\in[0,1/2]$ and $0$ for $p\in[1/2,1]$.}
\end{corollary}
\begin{proof}
Since we are restricting the depolarizing probabilities to the range $p\in(0,p_0^{d_l})$, we know that the superaddivity gain of the quantum capacity, $\xi_\mathrm{nc}(\Lambda_p^d)$, is a monotonically decreasing function with $d$ from Theorem \ref{thm:decrease}. Therefore, by taking into account the following chain of inequalities:
\begin{equation}
\begin{split}
\xi_\mathrm{nc}(\Lambda_p^d) & = Q_\mathrm{nc}^d(\Lambda_p^d) - Q_\mathrm{coh}^d(\Lambda_p^d)\geq P^d(\Lambda_p^d) - Q_\mathrm{coh}(\Lambda_p^d) \\& \geq P^d(\Lambda_p^d)- P_1(\Lambda_p^d) = \xi^P(\Lambda_p^d),
\end{split}
\end{equation}
and, therefore, the upper bound for the superadditivity gain of the quantum channel capacity upper bounds the superadditivity gain of the private capacity of qudit depolarizing channels too. Consequently, the possible room for increasing the achievable rate in a task of private classical communication over a qudit depolarizing channel decrases as the dimension of the system increases.

The second part of the corollary is straightforward from Corollary \ref{col:inftyCap} due to the fact that
\begin{equation}
\lim_{d\rightarrow\infty} \xi_\mathrm{nc}(\Lambda_p^d)=0,
\end{equation}
thus, $\xi^P(\Lambda_p^\infty) = 0$. This implies that $P^d(\Lambda_p^\infty) = \lim_{d\rightarrow\infty}Q_\mathrm{nc}(\Lambda_p^\infty) = 1-2p$, for $p\in[0,1/2]$. The complementary depolarizing parameter region is trivial from the fact that the channel is antidegradable and, thus, the private capacity vanishes.
\end{proof}

Note that the result imlying that the quantum channel capacity and the private capacity of the qudit depolarizing channels coincide when the dimension of the system is let to grow indefinitely is an interesting result since, at the time of writing, only the class of more capable (which includes the class of degradable channels) quantum channels present such equality \cite{watanabe,partialorders,singh}, while depolarizing channels are not more capable.

\section{Implications on quantum error correction with correlated depolarizing noise}\label{sec:implications}
As stated before, the $d$-dimensional depolarizing channel can be used to describe a noise map over a set of $n$ qubits for which the noise occurs in a very correlated manner. In this sense, this channel will have a dimension that is equal to the whole qubit system, i.e. $d=2^n$. In order to better understand the correlated noise model descirbed by the $d$-dimensional depolarizing channel for this systems, note that expression \eqref{eq:depol} can be rewritten as \cite{kathri}
\begin{equation}\label{eq:depolalt}
\Lambda_p^{d=2^n}(\rho) = (1-p+\frac{p}{2^{2n}})\rho + \frac{p}{2^{2n}}\sum_{\{\bar{j},\bar{k}\}\setminus\{\bar{0},\bar{0}\}} \mathrm{X}^{\bar{j}} \mathrm{Z}^{\bar{k}} \rho \mathrm{Z}^{\bar{k}}\mathrm{X}^{\bar{j}} ,
\end{equation}
where $\mathrm{X}^{\bar{j}} = \mathrm{X}^{j_1}\otimes\mathrm{X}^{j_2}\otimes\cdots\mathrm{X}^{j_n}$ and $\mathrm{Z}^{\bar{k}} = \mathrm{Z}^{k_1}\otimes\mathrm{Z}^{k_2}\otimes\cdots\mathrm{Z}^{k_n}$, and $\mathrm{X},\mathrm{Z}$ are the bit and phase flip Pauli matrices, respectively. By inspecting this expression, it can be observed that the $d$-dimensional depolarizing channel refers to a channel in which all the non-trivial Pauli elements of the $n$-fold Pauli group are applied in an equiprobable manner. Therefore, this channel represents a channel in which there exists a visible correlation in the Pauli errors that each of the qubits of the system experiences. A visual example of why this is said to be correlate can be seen in the fact that for this channel an error or weight $n$ would occur with a probability $p/2^{2n}$ while in an independent depolarizing channel, such probability would be gien by the product of the probability of error of the channel, i.e. $(p/4)^n$. Thus, such even is much more infrequent for the uncorrelated depolarizing channel.

Following this logic, consider for example a rotated planar surface code with distance $d=21$ \cite{surface} or a length $1000$ quantum turbo code \cite{qtc}. Those lengths refer to quantum error correction codes with a good performance. Note that for the dimensions of the whole system for such codes will be $d=2^{441}$ and $2^{1000}$, respectively. Thus, those codes have humongous dimensionality. Quantum error correction codes are presumed to operate over a large quantity of qubits, similarly as the examples provided, hence, if the noise experienced by those systems ha a signifcant correlation, that is, similar to the depolarizing channel presented in Eq. \eqref{eq:depolalt}; then there will no be possible superadditive effects and the optimal communication/correction rates will be achieved by random stabilizer codes.
\section{Conclusion}
In this article we have studied how the potential superadditivity effects of the quantum capacity of the qudit depolarizing channel relate to the dimension of the quantum systems in consideration. We proved that whenever the dimension of the $d$-dimensional depolarizig channel increases, the potential gain in terms of qudits per channel use decreases. This is an important result since it implies that for very high dimensional systems the channel coherent information and the quantum channel capacity will be very similar for the depolarizing channel, which results in the fact that random block codes on the typical subspace of the optimal input will be capacity achieving. We also observed that when $\infty$-dimensional or bosonic depolarizing channels are considered, the coherent information results to be an additive quantity, making the superadditivity gain to vanish for all depolarizing probabilities. We proved that the private capacity of qudit depolarizing channels behaves similarly in the sense that its potential superadditivity gain decrases with the dimension of the system. Asymptotically, the ability of sending private classical information over such family of channels is also an additive quantity and, interestingly, it coicides with the previously discussed quantum channel capacity. We also discussed the fact that since high dimensional depolarizing channels exhibit additive coherent information, the use of entangled input states is not required for optimal quantum information protection for such cases, significantly relaxing the resources required.

We have conducted this analysis of the reduction of superadditivity effects for depolarizing channels, but we consider that this type of arguments can be used in order to study how superadditivity behaves in high dimensions for other quantum channels. Similar proofs for other general qudit channels could be potentially obtained by squeezing upper bounds for their capacities with their coherent informations. Also, it is noteworthy to state that the Clifford twirl of a general $d$-dimensional channel results in a qudit depolarizing channel \cite{twirldep}, implying that since its capacity lower bounds the capacity of the original channel, the results obtained here may be somehow extended.  This way it could be concluded if the gain in qudits per channel use decreases with respect to the dimension for every quantum channel that admits a seamless extension to $d$-dimensions, implying that seeking such effects should be restricted for low dimensional quantum channels. Additionally, the behaviour of other channel capacities such as the Local Operations and Classical Communications (LOCC)-assisted quantum capacity \cite{lamiWilde,partialorders}, $Q_\leftrightarrow$, or the secret-key agreement capacity (LOCC-assisted private capacity) \cite{lamiWilde,partialorders}, $P_\leftrightarrow$, can also be studied for the family of depolarizing channels and for general maps too. These thoughts are conjectures and are deemed as future work.

\begin{acknowledgments}
This work was supported by the Spanish Ministry of Economy and Competitiveness through the ADELE (Grant No. PID2019-104958RB-C44) and MADDIE projects (Grant No. PID2022-137099NB-C44), by the Spanish Ministry of Science and Innovation through the proyect Few-qubit quantum hardware, algorithms and codes, on photonic and solid-state systems (PLEC2021-008251), by the Ministry of Economic Affairs and Digital Transformation of the Spanish Government through the QUANTUM ENIA project call - QUANTUM SPAIN project, and by the European Union through the Recovery, Transformation and Resilience Plan - NextGenerationEU within the framework of the Digital Spain 2025 Agenda.
\end{acknowledgments}

\end{document}